\documentclass[letterpaper, 10 pt,conference]{ieeeconf}  
\usepackage{amssymb}
\usepackage{amsmath,mathrsfs}
\usepackage{color}
\usepackage{graphicx}
\usepackage{epstopdf}
\usepackage{mathtools}
\usepackage{lineno,hyperref}
\usepackage{cleveref}
\modulolinenumbers[5]

\newtheorem{theorem}{Theorem}[section]

\newtheorem{remark}{Remark}[section]

\usepackage{multirow}
\usepackage{color}
\usepackage{cite}

\usepackage{fancyhdr}
\usepackage{times}
\usepackage{wrapfig}
\usepackage{pdfpages}
\usepackage{amsfonts}
\usepackage{enumerate}
\usepackage{amssymb}
\usepackage{stmaryrd}
\usepackage{float}

\providecommand{\keywords}[1]{\textbf{\textit{Index terms---}}}

\IEEEoverridecommandlockouts                              

\title{\LARGE \bf
An asymptotically optimal indirect approach to continuous-time system identification}

\author{Rodrigo A. \text{Gonz\'alez}, Cristian R. Rojas and James S. Welsh 
\thanks{This work was supported by the Swedish Research Council under contract number 2016-06079 (NewLEADS).}%
\thanks{R.A. Gonz\'alez and C.R. Rojas are with the Department of Automatic Control and ACCESS Linnaeus Centre, KTH Royal Institute of Technology, 10044 Stockholm, Sweden (e-mails: grodrigo@kth.se, crro@kth.se).}%
\thanks{James S. Welsh is with the School of Electrical Engineering and Computer Science, University of Newcastle, Australia (e-mail: james.welsh@newcastle.edu.au).}%
}
\begin{document}

\maketitle

\begin{abstract}
The indirect approach to continuous-time system identification consists in estimating continuous-time models by first determining an appropriate discrete-time model. For a zero-order hold sampling mechanism, this approach usually leads to a transfer function estimate with relative degree 1, independent of the relative degree of the strictly proper real system. In this paper, a refinement of these methods is developed. Inspired by indirect PEM, we propose a method that enforces a fixed relative degree in the continuous-time transfer function estimate, and show that the resulting estimator is consistent and asymptotically efficient. Extensive numerical simulations are put forward to show the performance of this estimator when contrasted with other indirect and direct methods for continuous-time system identification.
\end{abstract}

\section{Introduction}
System identification deals with the problem of estimating adequate models of dynamical systems from input-output data. The methods developed over the years in this field have seen applications in many areas of science and engineering, and comprehensive literature has been written on the subject \cite{soderstrom1988system,ljung1998system,pintelon2012system}. 

When postulating a mathematical model for describing a dynamical system based on sampled data, one must decide between obtaining a discrete-time (DT), or a continuous-time (CT) model. In a predominantly digital era, DT system identification has been studied thoroughly (see, e.g., \cite{soderstrom1988system,ljung1998system}). Nevertheless, interest in CT models still persists due to its advantages over discrete-time. For example, grey-box modelling \cite{bohlin2006practical}, which is commonly based on physical principles and conservation laws, is naturally suited for continuous time, as the parameters can usually be better interpreted in this domain. Also, CT models are known to have more intuitive dynamics, and they do not depend on a sampling period.

In CT system identification there are two main approaches, namely the direct and indirect approaches. For direct CT system identification, a CT model is obtained directly from the sampled data. The main difficulty present in the direct methods is the handling of derivatives, as they are not immediately available from discrete data points without amplifying noise \cite{rao2006identification}. 
To effectively deal with this issue, many well known methods have been proposed \cite{Garnier2008book}, with success in real applications \cite{garnier2015direct}. On the other hand, indirect methods for CT modelling first determine a suitable DT model via DT system identification methods like the Prediction Error Methods (PEM) or Maximum Likelihood (ML) \cite{ljung1998system}, and then transform this model into a CT equivalent model. Evidence has been shown regarding the advantages of direct over indirect CT model identification \cite{rao2002numerical}, although with a precise initialisation of PEM, the approaches seem comparable for certain sampling periods \cite{ljung2009experiments}.

Even though the indirect approach seems easy to implement, as there is much theory and literature concerning DT system identification, there are reasons why this approach is not always recommended. First, it may suffer from numerical inaccuracies at fast sampling, and requires a precise initialisation. In addition, it is not possible to select the desired numerator order of the CT model, as the estimated DT model will generally lead to a CT model with relative degree 1 in the case of sampling by a zero-order hold mechanism. Hence, an unnecessarily complex model structure is indirectly being estimated, which leads to a loss in accuracy according to the parsimony principle \cite{soderstrom1988system}. 

In this paper, we introduce a method that optimally imposes a desired relative degree in the indirect approach to continuous-time system identification. Based on Indirect PEM \cite{soderstrom1991indirect}, we prove that the proposed estimator is a consistent and asymptotically efficient estimator of the system's true parameter vector. Extensive numerical simulations show that the new method imposes the correct relative degree, while improving the statistical properties of the transfer function estimate, and achieves a performance that compares favourably against both standard direct and indirect approaches. 

The remainder of this paper is organised as follows. In Section \ref{sec2} the problem is formulated. Section \ref{sec3} provides an introduction to the indirect approach for CT system identification. In Section \ref{sec4} we derive an estimator that optimally enforces the desired relative degree for the indirect approach, and determine its properties. Section \ref{sec5} illustrates the method with extensive numerical examples. Finally, conclusions are drawn in Section \ref{sec6}.

\section{Problem formulation}
\label{sec2}
Consider a linear time-invariant, causal, stable, single input single output, CT system
\begin{flalign}
y(t) &= G_0(\rho)u(t) \notag \\
\label{eq1}
&= \frac{\beta_{n-r} \rho^{n-r}+\beta_{n-r-1} \rho^{n-r-1}+\dots+\beta_1 \rho + \beta_0}{\rho^n+\alpha_{n-1} \rho^{n-1}+\dots+\alpha_1 \rho+ \alpha_0}u(t),
\end{flalign}
where $\rho$ is the Heavyside operator, i.e., $\rho g(t) = \text{d}g(t)/\text{d}t$, and $r$ is the relative degree of the system. 
In this paper, we denote $\theta_c^0:= [\beta_{n-r} \hspace{0.15cm} \dots \beta_0 \hspace{0.15cm} \alpha_{n-1} \hspace{0.15cm} \dots \hspace{0.15cm} \alpha_0]^\top$ as the true CT system parameter vector. 

Suppose that the input-output signals are sampled with period $h$ and the resulting output is contaminated by an additive zero-mean white noise sequence $\{e(kh)\}_{k\in \mathbb{N}}$ of variance $\sigma^2$. That is,
\begin{equation}
\label{eq2}
y_m(kh) = y(kh)+e(kh), \quad k\in \mathbb{N}.
\end{equation}

%

The goal of CT system identification is to obtain a CT transfer function estimate for $G_0(\rho)$, given $N$ discrete input-output data measurements $\{u(kh),y_m(kh)\}_{k=1}^N$ and knowledge about the physical characteristics of the system, or the intersample behaviour. 
In this paper, we assume that the input is a piecewise constant signal between samples (i.e., zero-order hold behaviour \cite{astroem1984computer}).

For obtaining a model of $G_0(\rho)$, a simple way to proceed is to identify the zero-order hold equivalent model given the input and output data measurements by using standard PEM in the DT domain, and then return to the CT domain via zero-order hold equivalences. Although this procedure has good statistical properties, it does not impose relative degree constraints in the CT domain. Our goal is to optimally impose this constraint, which should lead to a statistically improved estimate of $G(s)$.

\section{The indirect approach to continuous-time system identification}
\label{sec3}
One approach to identifying a CT system is to first estimate the DT model given the input and output data samples, and then translate this model into continuous time. This is called the indirect approach, since it relies on DT system identification theory instead of obtaining immediately a CT model using CT system identification methods.

Much literature has been written regarding the first step of the indirect approach \cite{soderstrom1988system,ljung1998system}. The theoretically optimal solution is to apply the maximum likelihood method. This method is known to give consistent and asymptotically efficient estimates under very general conditions. Under the assumption that the additive white noise is Gaussian, the ML method is equivalent to PEM, which is one of the most celebrated parametric methods, and is available in the MATLAB System Identification Toolbox \cite{ljung2008system}. 

If a CT model for \eqref{eq1} is required, we should propose a DT model structure of the form  
\begin{equation}
\label{hz}
H(z)=\frac{b_{n-1} z^{n-1}+b_{n-2} z^{n-2}+\dots+b_1 z + b_0}{z^n+a_{n-1} z^{n-1}+\dots+a_1 z+ a_0}.
\end{equation}
If we define $\theta_d = [
b_{n-1} \dots  b_0 \hspace{0.15cm} a_{n-1} \dots a_{0}]^\top$ and denote $\hat{y}$ as the model's output, then the ML estimate is
\begin{equation}
\hat{\theta}_d = \arg \min_{\theta_d} \frac{1}{\sigma^2} \sum_{k=1}^N \|y_m(kh)-\hat{y}(kh,\theta_d) \|^2. \notag
\end{equation}
The next step is to transform this DT transfer function into an adequate CT model. This can be done in several ways. For example, the well-known Tustin transformation can be applied on the DT transfer function estimate by letting
\begin{equation}
z = \frac{sh+2}{2-sh}, \notag
\end{equation}
as reported in \cite{unbehauen1990continuous}. If it is assumed that the CT input signal is piecewise constant, the most natural mapping (used in e.g. \cite{sinha2000identification,ljung2009experiments}), is the zero-order hold sampling equivalence
\begin{equation}
\label{zoh}
H_0(z)=(1-z^{-1})\mathcal{Z}\left\{\mathcal{L}^{-1}\left\{\frac{G_0(s)}{s}\right\}\bigg\rvert_{t=kh} \right\},
\end{equation}
where $\mathcal{Z}$ and $\mathcal{L}$ denote the Z and Laplace transforms, respectively. This mapping is known to be a troublesome part of the indirect approach, as it can be ill conditioned and its uniqueness depends on a correct choice of the sampling period. 

In both of these mappings, the resulting CT model can have numerator parameters that exceed the desired numerator orders. This is generally the case when the relative degree of $G_0(s)$ is greater than 1. This problem contributes to poor accuracy and high standard deviations at high frequencies. One very simple way of treating this issue is setting to zero the numerator coefficients which should be zero, but this is not the best way of taking care of the information \cite{ljung2009experiments}. 

\section{Optimal enforcement of relative degree}
\label{sec4}
In this section we develop an indirect-approach estimator for the CT parameter vector $\theta_c^0$ that renders a CT transfer function estimate $\hat{G}(s)$ of a desired relative degree $r$. For this matter, we first focus on the PEM estimator of the zero-order hold equivalent model of $G_0(s)$.

For simplicity, we assume that the correct model order has been found. A model with structure \eqref{hz} is obtained by PEM, and the covariance matrix of $\hat{\theta}_d$ is also estimated. We know that the CT zero-order hold equivalent of this estimated model is in general given by
\begin{equation}
\hat{G}(s) = \frac{\hat{\beta}_{n-1} s^{n-1}+\hat{\beta}_{n-2} s^{n-2}+\dots+\hat{\beta}_1 s + \hat{\beta}_0}{s^n+\hat{\alpha}_{n-1} s^{n-1}+\dots+\hat{\alpha}_1 s+ \hat{\alpha}_0}. \notag
\end{equation} 
Define $\hat{\theta}_c = [\hat{\beta}_{n-1} \hspace{0.15cm} \dots \hat{\beta}_0 \hspace{0.15cm} \hat{\alpha}_{n-1} \hspace{0.15cm} \dots \hspace{0.15cm} \hat{\alpha}_{0}]^\top$, and denote by $\theta_d^0$ the true DT parameter vector. The parameters in $\hat{\theta}_c$ are related to $\hat{\theta}_d$ by the zero-order hold equivalence equations that can be derived by using \eqref{zoh} and comparing coefficients. This relation is a nonlinear mapping $f: \hat{\theta}_c \to f(\hat{\theta}_c) = \hat{\theta}_d$, which is differentiable almost everywhere. Hence, the following asymptotic relationship is valid for the covariance matrices of $\hat{\theta}_d$ and $\hat{\theta}_c$:
\begin{align}
\Sigma_{\hat{\theta}_d} &= \text{E}\{(\hat{\theta}_d-\theta_d^0)(\hat{\theta}_d-\theta_d^0)^\top\} \notag \\
&\approx \text{E}\{J(\hat{\theta}_c-\theta_c^0)(\hat{\theta}_c-\theta_c^0)^\top J^\top\} \notag \\
\label{jacobiantrick}
& = J \Sigma_{\hat{\theta}_c} J^\top,
\end{align}
where $J$ is the Jacobian matrix of $f$ evaluated at the naive estimation of $\theta_c^0$, that is, throwing away the high order coefficients of the numerator of $\hat{G}$ which should be zero\footnote{Note that the standard PEM estimate could have also been used for this matter, and the asymptotic relation still holds.}. 

Now, we propose to find an appropriate projection of $\hat{\theta}_c$ into a proper subspace of the parameter space that yields the desired relative degree. This subspace is simply the one formed by all vectors with first $r-1$ elements set to zero. Hence, we decide to study the following problem:
\begin{align}
\label{lagrange1}
\tilde{\theta}_c = \arg &\min_{\theta} \frac{1}{2} (\hat{\theta}_c- \theta)^\top \Sigma_{\hat{\theta}_c}^{-1} (\hat{\theta}_c- \theta) \\
\label{lagrange2}
&\text{s.t. \hspace{0.2cm}} \begin{bmatrix}
\textnormal{I}_{r-1} & 0 
\end{bmatrix}\theta = 0,
\end{align}
where $\textnormal{I}_{r-1}$ is the identity matrix of dimension $r-1$, $0$ is the null matrix of appropriate dimensions, and $\Sigma_{\hat{\theta}_c}^{-1}=J^\top \Sigma_{\hat{\theta}_d}^{-1}J$. The optimisation problem in this context can be interpreted as an application of the Indirect PEM \cite{soderstrom1991indirect}.

By Lagrange multiplier theory, the optimization problem in \eqref{lagrange1} is equivalent to calculating, for a suitable $\lambda$,
\begin{equation}
\tilde{\theta}_c = \arg \min_{\theta} \frac{1}{2}(\hat{\theta}_c- \theta)^\top \Sigma_{\hat{\theta}_c}^{-1} (\hat{\theta}_c- \theta) + \lambda^\top \begin{bmatrix}
\textnormal{I}_{r-1} & 0 
\end{bmatrix}\theta. \notag
\end{equation}
Partitioning $\Sigma_{\hat{\theta}_c}$ appropriately (dropping the subindex for simplicity), and differentiating with respect to $\theta$ we obtain
\begin{equation}
\label{prebla}
\tilde{\theta}_c = \hat{\theta}_c - \begin{bmatrix}
\Sigma_{11} & \Sigma_{12} \\
\Sigma_{21} & \Sigma_{22} 
\end{bmatrix} 
\begin{bmatrix}
\textnormal{I}_{r-1} \\
0 
\end{bmatrix} \lambda,
\end{equation}
Imposing \eqref{lagrange2} we obtain $\lambda = \Sigma_{11}^{-1} \begin{bmatrix}
\textnormal{I}_{r-1} & 0 
\end{bmatrix} \hat{\theta}_c$. If we denote by $C$ the Cholesky factorization matrix of $\Sigma_{\hat{\theta}_c}$ \cite{Horn2012} (i.e., a lower triangular matrix with positive diagonal entries such that $\Sigma_{\hat{\theta}_c} = CC^\top$) we can write \eqref{prebla} as 
\begin{align}
\tilde{\theta}_c &= \hat{\theta}_c - \Sigma 
\begin{bmatrix}
\textnormal{I}_{r-1} \\
0
\end{bmatrix} \Sigma_{11}^{-1} 
\begin{bmatrix}
\textnormal{I}_{r-1} & 0 
\end{bmatrix}
\hat{\theta}_c \notag \\
&= \hat{\theta}_c - \begin{bmatrix}
C_{11} & 0 \\
C_{21} & C_{22} 
\end{bmatrix} 
\begin{bmatrix}
\textnormal{I}_{r-1} \\
0
\end{bmatrix} C_{11}^{-1} 
\begin{bmatrix}
\textnormal{I}_{r-1} & 0 
\end{bmatrix}
\hat{\theta}_c \notag \\
\label{bla}
&= C
\begin{bmatrix}
0_{r-1}&  0^\top \\
0 & \textnormal{I}_{2n-r+1}
\end{bmatrix}
C^{-1} \hat{\theta}_c.
\end{align}
That is, the estimator \eqref{bla} can be seen as an $\mathcal{L}^2$ best approximation to the PEM CT estimate $\hat{\theta}_c$ that imposes the desired relative degree.

\subsection{Properties}
We briefly present the most important properties of estimator \eqref{bla} in the following theorems.

\begin{theorem}
\label{theorem1}
Consider the system described by \eqref{eq1} and \eqref{eq2}, where $\{e(kh)\}_{k=1}^N$ is a Gaussian white noise sequence. Assume that the sampling frequency $2\pi/h$ is larger than twice the largest imaginary part of the $s$-domain poles and there is no delay in the real system\footnote{These conditions can be relaxed, as long as the sampling frequency is such that the $z\to s$ transformation is well defined.}. Then, the estimator \eqref{bla} is a consistent and asymptotically efficient estimator of the real vector parameter $\theta_c^0$, provided the DT model set (with the chosen relative degree) contains the real system.
\end{theorem}
\begin{proof}
Under the Gaussian noise assumption, the DT PEM estimate can be interpreted as the ML estimate. Under the proposed sampling frequency, the $z\to s$ transformation is unique for sufficiently large $N$ \cite{kollar1996equivalence}. Hence, by the invariance principle of ML estimators \cite{zehna1966invariance}, the CT equivalent of the system's parameters is also an ML estimate. 

To prove the theorem, we only require that the assumptions in \cite{soderstrom1991indirect} are satisfied in this scenario, and then directly apply the results obtained in the cited contribution. First, note that the model structure given by this procedure contains the models with the desired relative degree (and the contention is proper if $r>1$), with a linear mapping between parameter vectors given by the matrix
\begin{equation}
T := \begin{bmatrix}
0_{r-1} & \textnormal{I}_{2n-r+1}
\end{bmatrix}. \notag
\end{equation}
Furthermore, provided that the DT model set contains the real system, both structures give parameter identifiability. Also note that $\Sigma_{\hat{\theta}_c}$, obtained via \eqref{jacobiantrick}, is a consistent estimate of the covariance matrix of $\hat{\theta}_c$. Hence, the results in \cite{soderstrom1991indirect} follow. Namely, the normalized estimation errors\footnote{For simplicity, we assume that the vector $\theta_c^0$ has the appropriate dimension, where zeros have been considered in the first terms if necessary.} $\sqrt{N}(\hat{\theta}_c-\theta_c^0)/\sigma$ and $\sqrt{N}(\tilde{\theta}_c-\theta_c^0)/\sigma$ are asymptotically normally distributed with zero means and their asymptotic covariance matrices satisfy the relation
\begin{equation}
\label{covariances}
[\textnormal{AsCov}(\tilde{\theta}_c-\theta_c^0)]^{-1} = T[\textnormal{AsCov}(\hat{\theta}_c-\theta_c^0)]^{-1}T^\top. 
\end{equation}
Moreover, following the steps in \cite[Section 3]{soderstrom1991indirect}, the improved PEM estimate is $\sqrt{N}-$consistent and has the same asymptotic distribution as $\hat{\theta}_c$, thereby proving its asymptotic efficiency.
\end{proof}
\begin{remark}

Note that by \eqref{bla} and \eqref{covariances}, the asymptotic covariances can be shown to satisfy the following properties:
\begin{flalign*}
&\bullet \textnormal{AsCov}(\tilde{\theta}_c-\theta_c^0,\hat{\theta}_c-\tilde{\theta}_c)=0,& \notag \\
&\bullet \textnormal{AsCov}(\tilde{\theta}_c-\theta_c^0)=\textnormal{AsCov}(\hat{\theta}_c-\theta_c^0) - \textnormal{AsCov}(\hat{\theta}_c-\tilde{\theta}_c).& \notag 
\end{flalign*}
Both of these claims follow by applying properties 10.5 and 10.6 from \cite[Chapter 10]{gourieroux1995statistics} to this context. These properties imply that for sufficiently large $N$, the proposed estimator can only decrease the covariance of the estimated parameters compared to standard PEM. The asymptotic covariances satisfy a Pythagorean relation, as the PEM estimate is projected orthogonally on to the proper subspace where $\theta_c^0$ lies. 
\end{remark}

Next, we establish that imposing a larger relative degree improves the accuracy of the estimates, provided that the highest relative degree model structure contains the real system.
\begin{theorem}
Given a plant $G_0(s)$ of order $n$ and relative degree $r>1$, consider CT candidate models of relative degree $r_1$ and $r_2$ and their improved PEM parameter vector estimates ${\tilde{\theta}_c}^{r_1}$ and ${\tilde{\theta}_c}^{r_2}$ respectively. If $r_1<r_2\leq r$, then their asymptotic covariance matrices satisfy $\textnormal{AsCov}(\tilde{\theta}_c^{r_2}) \preceq \textnormal{AsCov}(\tilde{\theta}_c^{r_1})$.
\end{theorem} 
\begin{proof}
The proof follows by applying \cite[Theorem 2]{gonzalez2017optimal}. Details are omitted due to length restrictions.

\end{proof}

\begin{remark}
The relative degree of the CT system is not always known. In some practical applications, physical knowledge about the system can give intuition. In addition, statistical measures can be used such as the coefficient of determination \cite{garnier2015direct}, or the Young Information Criterion \cite{young1989recursive}.
\end{remark}

\section{Monte Carlo simulation studies}
\label{sec5}
We will now study the performance of the proposed estimator under a series of experiments.

The system considered is the Rao-Garnier system \cite{rao2002numerical}, which is a linear fourth-order non-minimum phase system with complex poles that has been tested in many publications (see e.g. \cite{ljung2003initialisation,welsh2009continuous,garnier2015direct}) of CT system identification:
\begin{equation}
G(s) = \frac{-6400s+1600}{s^4+5s^3 + 408s^2 + 416s+1600}. \notag
\end{equation}
This system is interesting since it has two damped oscillatory modes at 2 and 20 rad/sec with damping of $0.1$ and $0.25$ respectively, and has a non-minimum phase zero at $s=0.25$. It is known that this is a particularly difficult system to estimate by PEM/ML, since these methods may converge to a local minimum if they are not well initialised \cite{ljung2003initialisation}.

Three methods have been compared: PEM, PEM with relative degree enforcement (labeled PEMrd \eqref{bla}), and the simplified refined IV method for continuous-time systems (SRIVC), which is one of the most successful direct methods available, and has been suggested for general use in one of the most recent surveys on CT system identification \cite{garnier2015direct}. Each method has been tested in $M$ different Monte Carlo simulations, and they have been evaluated according to the average normalized square error of the system estimate
\begin{equation}
\textnormal{MSE} \hspace{0.1cm}\hat{G} = \frac{1}{M} \sum_{i=1}^M \frac{\|\hat{G}_i-G_0\|_2^2}{\|G_0\|_2^2}, \notag
\end{equation}
the average normalized square error of the parameter vectors
\begin{equation}
\textnormal{MSE}\hspace{0.1cm} \hat{\theta} = \frac{1}{M} \sum_{i=1}^M \frac{\|\hat{\theta}_c^i-\theta_c^0\|_2^2}{\|\theta_c^0\|_2^2}, \notag
\end{equation}
and the fit measure
\begin{equation}
\textnormal{Fit} = \frac{100}{M}\sum_{i=1}^M \left[1- \frac{\|\hat{y}^i-y\|_2}{\|y-\bar{y}\|_2} \right], \notag
\end{equation}
where $y$ is the noise-free output sequence (the simulated data without the additive measurement noise), $\hat{y}^i$ is the simulated output sequence of the $i$-th estimated model, and $\bar{y}$ is the average value of $\{y(kh)\}_{k=1}^N$.

We have run PEM using the standard MATLAB System Identification Toolbox \cite{ljung2008system} with the \texttt{oe} command, and assumed that the correct order of the system is known. The search algorithm has been initialised with the estimate given by the Null Space Fitting method\footnote{PEM initialised with the estimate from SRIVC (as in \cite{ljung2009experiments}) has also been tested with similar results.} \cite{galrinho2017weighted}. We based PEMrd on the PEM estimate previously obtained. The required Jacobian matrix has been numerically calculated via finite differences, and the correct relative degree has been imposed for this estimator. We have used the command \texttt{d2c} of MATLAB in both cases. SRIVC has been implemented with the CONTSID toolbox for MATLAB \cite[Chapter 9]{Garnier2008book} with default initialisation, and has been set to estimate the model
\begin{equation}
\hat{G}(s) = \frac{\beta_1 s + \beta_0}{s^4 + \alpha_3 s^3+\alpha_2 s^2+\alpha_1 s+\alpha_0}. \notag
\end{equation}

\subsection{Effect of the number of data points and sampling rate}
\label{subsec1}
In this study, we have designed the input as a pseudorandom binary sequence (PRBS) of amplitude switches between $0$ and $2$. For the first input sequence, the number of stages of the shift register is $n=10$ and the data length of the shortest interval is $p=7$. Hence, a sequence of $N=7161$ data points has been obtained. The noise is a zero-mean white Gaussian noise signal, where the variance has been set such that the signal-to-noise ratio (SNR in dB) between the noiseless output sequence and the noise equals $10$ dB.

Three Monte Carlo studies have been performed with $M=500$ simulations of different noise realisations each one, for $h=0.01, 0.05, 0.1$. The results are shown in Table \ref{table1}. 
\begin{table}[h!]
\centering
\caption{Monte Carlo simulation results for PRBS input of total length $N=7161$.}
\label{table1}
\begin{tabular}{|c||c|c|c|c|}
\hline
$h$                    & Method & $\textnormal{MSE} \hspace{0.1cm} \hat{G}$ & $\textnormal{MSE} \hspace{0.1cm} \hat{\theta}$ & Fit     \\ \hline
\hline
\multirow{3}{*}{0.01} & PEM                         & $1.113\cdot 10^{-4}$              & $6.757\cdot 10^{-5}$                        & 98.9742 \\
                      & PEMrd                       & $0.732\cdot 10^{-4}$              & $4.283\cdot 10^{-5}$                        & 99.1219 \\
                      & SRIVC                       & $0.733\cdot 10^{-4}$              & $4.292\cdot 10^{-5}$                        & 99.1217 \\ \hline
\multirow{3}{*}{0.05} & PEM                         & $1.996\cdot 10^{-3}$              & $5.423\cdot 10^{-4}$                        & 98.9645 \\
                      & PEMrd                       & $1.406\cdot 10^{-3}$              & $3.933\cdot 10^{-4}$                        & 99.1141 \\
                      & SRIVC                       & $1.397\cdot 10^{-3}$              & $3.925\cdot 10^{-4}$                        & 99.1146 \\ \hline
\multirow{3}{*}{0.1}  & PEM                         & $3.275\cdot 10^{-3}$              & $7.937\cdot 10^{-4}$                        & 98.9436 \\
                      & PEMrd                       & $1.914\cdot 10^{-3}$              & $4.797\cdot 10^{-4}$                        & 99.0884 \\
                      & SRIVC                       & $1.922\cdot 10^{-3}$              & $4.749\cdot 10^{-4}$                        & 99.0870 \\ \hline
\end{tabular}
\end{table}

In order to analyse the performance of each estimator under less data, we have set the number of stages to $n=9$ and the data length of the shortest interval to $p=3$, resulting in a input of $N=1533$ data points. With the same SNR as the test above, the results for $500$ Monte Carlo simulations for each sampling period can be found in Table \ref{table2}.
\begin{table}[h!]
\centering
\caption{Monte Carlo simulation results for PRBS input of total length $N=1533$.}
\label{table2}
\begin{tabular}{|c||c|c|c|c|}
\hline

$h$                    & Method & $\textnormal{MSE} \hspace{0.1cm} \hat{G}$ & $\textnormal{MSE} \hspace{0.1cm} \hat{\theta}$ & Fit     \\ \hline
\hline
\multirow{3}{*}{0.01} & PEM                         & $5.882\cdot 10^{-4}$              & $3.651\cdot 10^{-4}$                        & 97.7791 \\
                      & PEMrd                       & $4.025\cdot 10^{-4}$              & $2.057\cdot 10^{-4}$                        & 98.0705 \\
                      & SRIVC                       & $4.017\cdot 10^{-4}$              & $2.060\cdot 10^{-4}$                        & 98.0719 \\ \hline
\multirow{3}{*}{0.05} & PEM                         & $7.431\cdot 10^{-4}$              & $4.539\cdot 10^{-4}$                        & 97.7892 \\
                      & PEMrd                       & $4.316\cdot 10^{-4}$              & $2.788\cdot 10^{-4}$                        & 98.1172 \\
                      & SRIVC                       & $4.319\cdot 10^{-4}$              & $2.789\cdot 10^{-4}$                        & 98.1167 \\ \hline
\multirow{3}{*}{0.1}  & PEM                         & $1.959\cdot 10^{-2}$              & $2.915\cdot 10^{-3}$                        & 97.7416 \\
                      & PEMrd                       & $1.077\cdot 10^{-2}$              & $2.184\cdot 10^{-3}$                        & 98.0339 \\
                      & SRIVC                       & $9.992\cdot 10^{-3}$              & $3.877\cdot 10^{-3}$                        & 97.9754 \\ \hline
\end{tabular}
\end{table}

Both Tables \ref{table1} and \ref{table2} show that the refined PEM estimator statistically improves the estimates given by PEM, and is a very competitive method against SRIVC, even for high frequency sampling. Note that under less data points, PEMrd still outperforms PEM for every sampling period, which indicates that the asymptotic properties studied in Section \ref{sec4} can be observed in practical finite data cases as well.
\begin{remark}
In Tables \ref{table1} and \ref{table2} we have discarded cases where PEM has delivered estimates with one pole in the negative real axis. Fortunately this scenario is very uncommon (9 cases seen in 3000 simulations). A similar phenomenon was observed in Table \ref{table2} for $h=0.1$ in SRIVC, where 2 estimates gave negative fit values. These simulations were not considered either.
\end{remark}

\subsection{Multisine input}
\label{sec5multisine}
A different input signal has also been tested. We have taken 2000 data measurements of a multisine input given by the sum of sine waves of angular frequencies $\omega = 0.5,1,5,8,10,12,15,20,25,30$. The standard deviation of the additive noise has been set to 0.1, and the median of $500$ Monte Carlo simulations of the normalized model error, normalized parameter error, and fit have been obtained. The results are shown in Table \ref{table3}. 

\begin{table}[h!]
\centering
\caption{Monte Carlo simulation results for multisine wave input of total length $N=2000$.}
\label{table3}
\begin{tabular}{|c||c|c|c|c|}
\hline
$h$                    & Method & $\frac{\|\hat{G}-G_0\|^2}{\|G_0\|^2}$ &  $\frac{\|\hat{\theta}-\theta_0\|^2}{\|\theta_0\|^2}$ & Fit     \\ \hline
\hline
\multirow{3}{*}{0.01} & PEM                         & $8.799\cdot 10^{-5}$              & $7.269\cdot 10^{-5}$                        & 99.4319 \\
                      & PEMrd                       & $1.352\cdot 10^{-5}$              & $4.435\cdot 10^{-6}$                        & 99.8173 \\
                      & SRIVC                       & $1.439\cdot 10^{-1}$              & $6.205\cdot 10^{-2}$                        & 80.8446 \\ \hline
\multirow{3}{*}{0.02} & PEM                         & $1.422\cdot 10^{-4}$              & $8.368\cdot 10^{-5}$                        & 99.2684 \\
                      & PEMrd                       & $2.294\cdot 10^{-5}$              & $1.013\cdot 10^{-5}$                        & 99.7325 \\
                      & SRIVC                       & $2.141\cdot 10^{-2}$              & $4.599\cdot 10^{-3}$                        & 94.3144 \\ \hline
\end{tabular}
\end{table}
For this input and sampling method, SRIVC performs poorly, while PEM and PEMrd normally reach the global optimum. Even though the fit is very near the optimal, PEMrd is consistently better than standard PEM at all metrics.

\subsection{Mean and covariance of the estimated parameters}
As established in Theorem \ref{theorem1}, the improved PEM estimate should reduce (at least asymptotically) the covariance of the parameter vector. To test this, we have obtained the mean and standard deviation of each parameter given by the Monte Carlo study under the setup of Section \ref{subsec1} with $h=0.05$ and $N=7161$. The results are shown in Table \ref{table4}. It can be observed that the mean values are similar in all methods (except standard PEM, which does not estimate the correct model structure), and PEMrd provides the lowest standard deviation for every parameter.

\begin{table*}[t!]
\centering
\caption{Estimated parameter value means and standard deviations for each method, considering $h=0.05$.}
\label{table4}
\begin{tabular}{|c|c|c|c|c|c|c|c|c|c|}
\hline
\multirow{2}{*}{Method} & \multirow{2}{*}{\begin{tabular}[c]{@{}c@{}}Parameter\\ True value\end{tabular}} & \multirow{2}{*}{\begin{tabular}[c]{@{}c@{}}$b_1$\\ 0\end{tabular}}      & \multirow{2}{*}{\begin{tabular}[c]{@{}c@{}}$b_2$\\ 0\end{tabular}}      & \multirow{2}{*}{\begin{tabular}[c]{@{}c@{}}$b_3$\\ -6400\end{tabular}}     & \multirow{2}{*}{\begin{tabular}[c]{@{}c@{}}$b_4$\\ 1600\end{tabular}}    & \multirow{2}{*}{\begin{tabular}[c]{@{}c@{}}$a_1$\\ 5\end{tabular}}     & \multirow{2}{*}{\begin{tabular}[c]{@{}c@{}}$a_2$\\ 408\end{tabular}}   & \multirow{2}{*}{\begin{tabular}[c]{@{}c@{}}$a_3$\\ 416\end{tabular}}    & \multirow{2}{*}{\begin{tabular}[c]{@{}c@{}}$a_4$\\ 1600\end{tabular}}    \\
                        &                                                                                 &                                                                         &                                                                         &                                                                            &                                                                          &                                                                        &                                                                        &                                                                         &                                                                          \\ \hline
\multirow{2}{*}{PEM}    & \multirow{2}{*}{\begin{tabular}[c]{@{}c@{}}Mean\\ Std. Dev\end{tabular}}        & \multirow{2}{*}{\begin{tabular}[c]{@{}c@{}}-0.044\\ 0.963\end{tabular}} & \multirow{2}{*}{\begin{tabular}[c]{@{}c@{}}0.301\\ 11.414\end{tabular}} & \multirow{2}{*}{\begin{tabular}[c]{@{}c@{}}-6403.01\\ 147.13\end{tabular}} & \multirow{2}{*}{\begin{tabular}[c]{@{}c@{}}1601.33\\ 47.85\end{tabular}} & \multirow{2}{*}{\begin{tabular}[c]{@{}c@{}}5.006\\ 0.399\end{tabular}} & \multirow{2}{*}{\begin{tabular}[c]{@{}c@{}}408.16\\ 7.98\end{tabular}} & \multirow{2}{*}{\begin{tabular}[c]{@{}c@{}}416.25 \\ 9.27\end{tabular}} & \multirow{2}{*}{\begin{tabular}[c]{@{}c@{}}1600.49\\ 33.29\end{tabular}} \\
                        &                                                                                 &                                                                         &                                                                         &                                                                            &                                                                          &                                                                        &                                                                        &                                                                         &                                                                          \\ \hline
\multirow{2}{*}{PEMrd}  & \multirow{2}{*}{\begin{tabular}[c]{@{}c@{}}Mean\\ Std. Dev.\end{tabular}}       & \multirow{2}{*}{\begin{tabular}[c]{@{}c@{}}0\\ 0\end{tabular}}          & \multirow{2}{*}{\begin{tabular}[c]{@{}c@{}}0\\ 0\end{tabular}}          & \multirow{2}{*}{\begin{tabular}[c]{@{}c@{}}-6397.25\\ 122.39\end{tabular}} & \multirow{2}{*}{\begin{tabular}[c]{@{}c@{}}1599.79\\ 42.39\end{tabular}} & \multirow{2}{*}{\begin{tabular}[c]{@{}c@{}}4.994\\ 0.315\end{tabular}} & \multirow{2}{*}{\begin{tabular}[c]{@{}c@{}}407.88\\ 7.11\end{tabular}} & \multirow{2}{*}{\begin{tabular}[c]{@{}c@{}}415.84\\ 8.33\end{tabular}}  & \multirow{2}{*}{\begin{tabular}[c]{@{}c@{}}1599.27\\ 28.59\end{tabular}} \\
                        &                                                                                 &                                                                         &                                                                         &                                                                            &                                                                          &                                                                        &                                                                        &                                                                         &                                                                          \\ \hline
\multirow{2}{*}{SRIVC}  & \multirow{2}{*}{\begin{tabular}[c]{@{}c@{}}Mean\\ Std. Dev.\end{tabular}}       & \multirow{2}{*}{\begin{tabular}[c]{@{}c@{}}0\\ 0\end{tabular}}          & \multirow{2}{*}{\begin{tabular}[c]{@{}c@{}}0\\ 0\end{tabular}}          & \multirow{2}{*}{\begin{tabular}[c]{@{}c@{}}-6399.19\\ 132.05\end{tabular}} & \multirow{2}{*}{\begin{tabular}[c]{@{}c@{}}1600.49\\ 44.21\end{tabular}} & \multirow{2}{*}{\begin{tabular}[c]{@{}c@{}}5.014\\ 0.338\end{tabular}} & \multirow{2}{*}{\begin{tabular}[c]{@{}c@{}}407.99\\ 7.75\end{tabular}} & \multirow{2}{*}{\begin{tabular}[c]{@{}c@{}}416.61\\ 8.98\end{tabular}}  & \multirow{2}{*}{\begin{tabular}[c]{@{}c@{}}1599.71\\ 31.1\end{tabular}}  \\
                        &                                                                                 &                                                                         &                                                                         &                                                                            &                                                                          &                                                                        &                                                                        &                                                                         &                                                                          \\ \hline
\end{tabular}
\end{table*}

\subsection{Direct comparison with standard PEM}
In this subsection we analyse the improvements of the novel method over the standard PEM estimates.

To show the impact of selecting and enforcing the correct relative degree, we have focused on the Bode Diagrams of the resulting models for PEM and PEMrd. In Figure \ref{fig1} we have plotted the frequency response estimates of 100 Monte Carlo simulations under the setup in Section \ref{subsec1}, for $h=0.05$, $N=7161$.

It is clear by Fig. \ref{fig1} that the improvement over PEM is mainly at high frequencies. This is intuitive, since the relative degree determines the asymptotic slope of the Bode diagram of magnitude. The proposed method enforces the true asymptotic slope, leading to an important gain in accuracy in both magnitude and phase. 
\begin{figure}[h!]
\centering{
\includegraphics[clip, trim=4.3cm 10.2cm 4cm 10.6cm, width=0.5\textwidth]{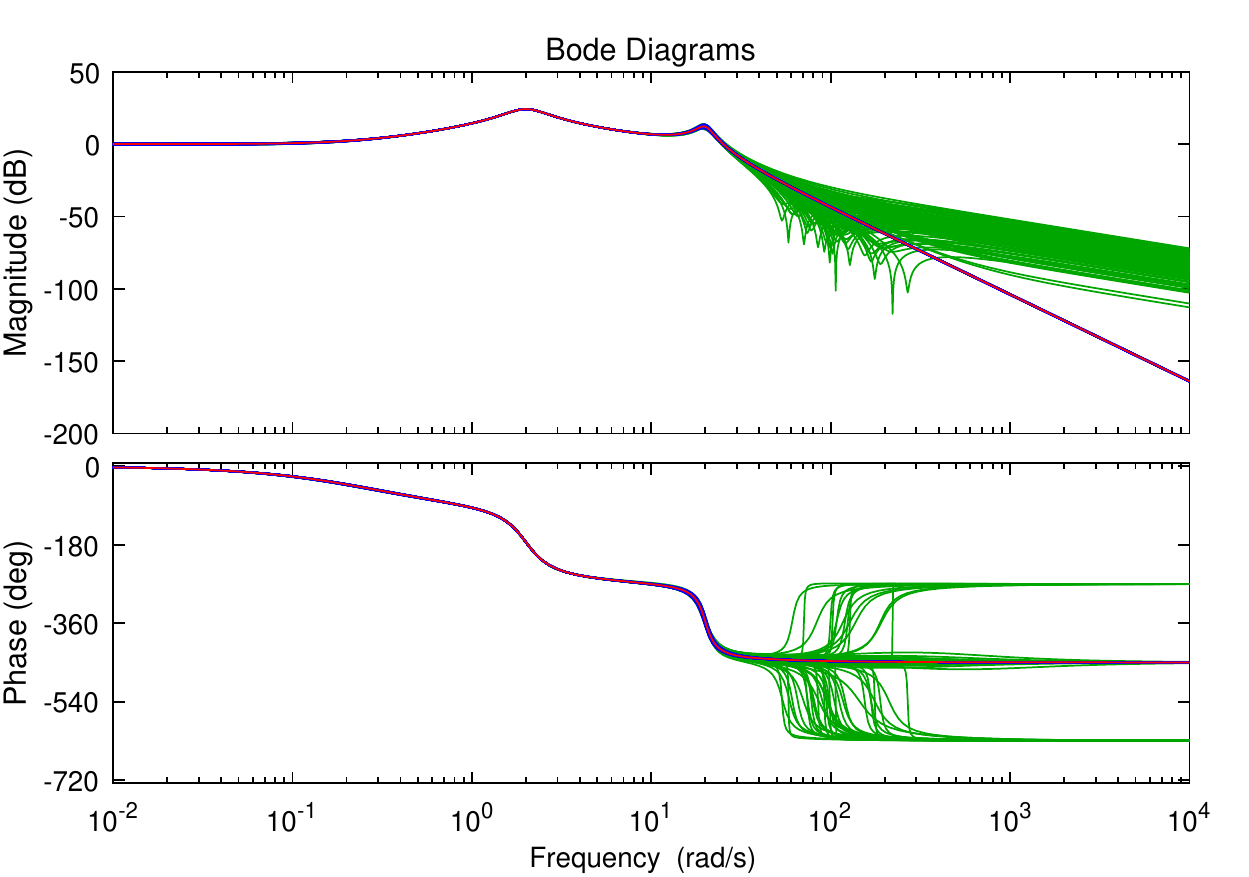}
\caption{Bode Diagram of 100 estimates by PEM (green), 100 estimates by PEMrd (blue), and the real system (red).}
\label{fig1}}
\end{figure}

Also, direct comparison plots between PEM and PEMrd are shown in Figure \ref{fig2}. These plots compare the normalized model and parameter error, and the fit of PEM and PEMrd for each Monte Carlo experiment (500 in total). PEMrd outperforms standard PEM in most Monte Carlo simulations, specially in the fit comparison, where 496 out of 500 experiments have lead to an increase in fit under PEM refinement.
\begin{figure}[h!]
\centering{
\includegraphics[trim=3.5cm 8cm 3.5cm 7.5cm, width=0.5\textwidth]{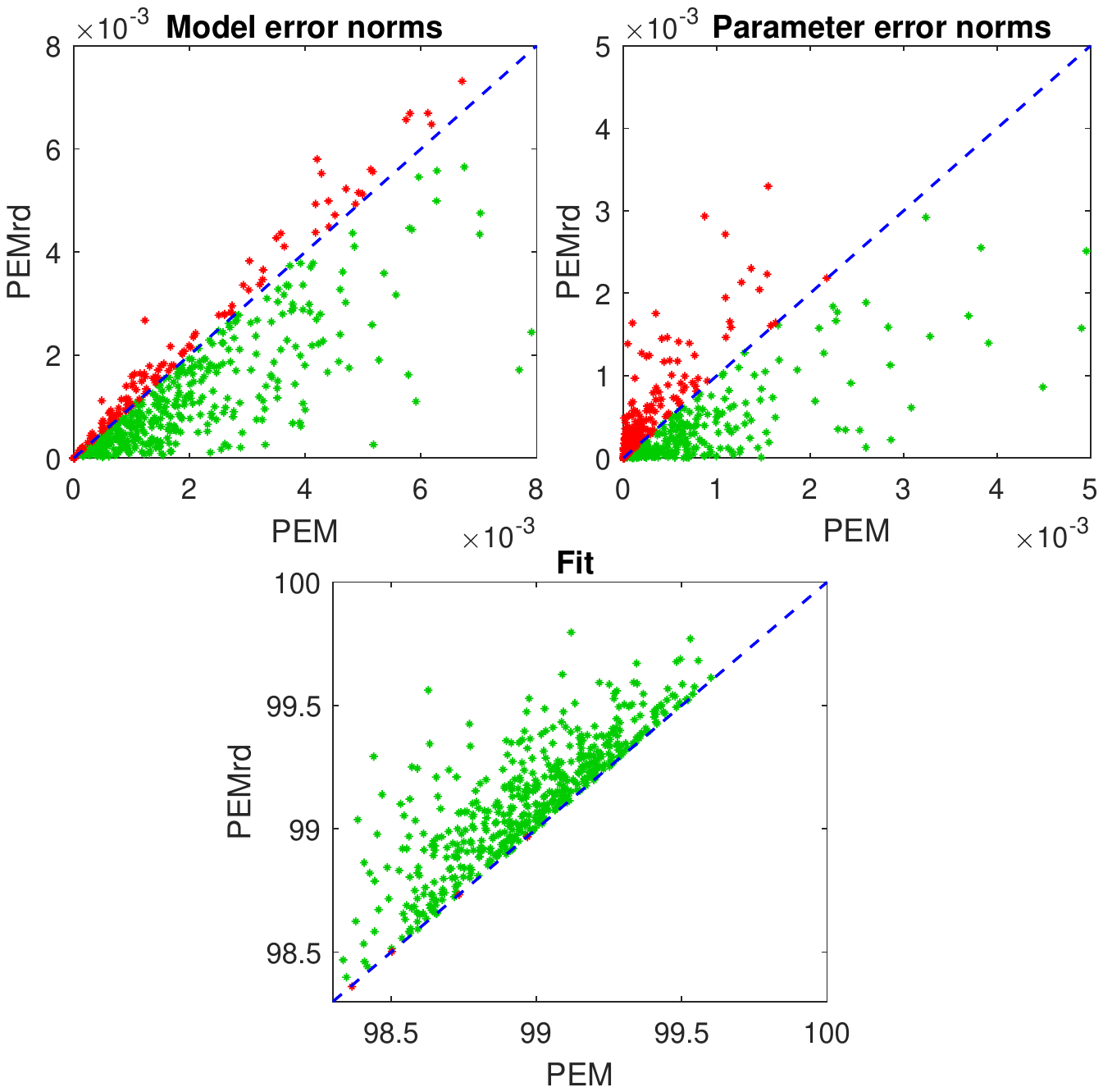}
\caption{Direct comparison plots between PEM and PEMrd for 500 Monte Carlo simulations. The green dots correspond to Monte Carlo simulations where PEMrd outperforms PEM. Red dots represent the opposite, and the dashed blue line is the separatrix.}
\label{fig2}}
\end{figure} 

\subsection{Random systems}
To obtain the average performance of each estimator, we have tested them on a data set created with 500 random systems of order 3 and relative degree 2, by using the \texttt{rss} command in MATLAB. The slowest pole has been set to have real part less than -0.1. The input was a unit variance Gaussian white noise, and the additive white noise was also Gaussian, of standard deviation equal to $5\%$ of the maximum value of the noiseless output. The sampling period has been chosen as 10 times faster than the fastest pole or zero of the real system.

We have computed the median of the metrics used above, as in Section \ref{sec5multisine}, for the 500 random systems. We also have counted failures of the estimators, which are the cases when the estimates produced a negative fit, when the algorithm crashed, or when it was not possible to correctly initialise the PEM estimate having tried reducing by a factor of 2 the sampling rate for initialisation estimation. 

The results can be seen in Table \ref{table5}. Although all estimators report failures, the PEMrd again shows promising results in all metrics. As pointed out in \cite{ljung2009experiments} and seen in these tests, initialisation aspects are in fact a major issue concerning the reliability of these algorithms.
\begin{table}[t!]
\centering
\caption{Monte Carlo simulation results for random systems.}
\label{table5}
\begin{tabular}{|c||c|c|c|c|}
\hline
Method & $\frac{\|\hat{G}-G_0\|^2}{\|G_0\|^2}$ & $\frac{\|\hat{\theta}-\theta_0\|^2}{\|\theta_0\|^2}$ & Fit     & \multicolumn{1}{l|}{N$^\circ$ Failures} \\ \hline
\hline
PEM    & $1.449\cdot 10^{-4}$                  & $3.328\cdot 10^{-3}$                                 & 98.871  & 12                           \\
PEMrd  & $1.046\cdot 10^{-4}$                  & $1.872\cdot 10^{-3}$                                 & 98.997  & 16                           \\
SRIVC  & $1.111\cdot 10^{-4}$                  & $2.126\cdot 10^{-3}$                                 & 98.9723 & 30            \\         
\hline     
\end{tabular}
\end{table}
\section{Conclusions}
\label{sec6}
We have proposed a refinement to the standard PEM estimator for indirect continuous-time system identification that achieves an asymptotically optimal desired relative degree enforcement. An explicit expression for this estimator has been found, and its statistical properties have been analysed. Extensive simulations using both standard benchmarks and random systems have been put forward with promising results. These have shown that the refinement to standard PEM leads to an important improvement in all the statistical metrics studied, and its performance is comparable, if not better, to SRIVC for all sampling periods in this study provided PEM is initialised correctly.

\bibliography{References}
\end{document}